\documentclass[12pt, journal, draftclsnofoot, onecolumn]{IEEEtran}
\usepackage{mathtools,enumitem,color,subfigure,bbm}
\usepackage{graphicx,tabularx,array}
\usepackage{hyperref}
\usepackage{amsthm,mathrsfs}
\usepackage{makecell}
\usepackage{multirow}
\usepackage[T1]{fontenc}
\usepackage{cite}

\usepackage[utf8]{inputenc} 
\usepackage[T1]{fontenc}
\usepackage{url}
\usepackage{ifthen}
\usepackage{enumitem}
\IEEEoverridecommandlockouts

\usepackage{color}
\usepackage{graphicx,tabularx,array,amsmath,amsthm,thmtools,caption}

\usepackage{mathtools}

\usepackage{amsfonts}
\usepackage{bm}
\usepackage{bbm}
\usepackage{makecell}
\usepackage{multirow}
 \usepackage{amssymb}
\usepackage{txfonts}
\usepackage[T1]{fontenc}
\usepackage{tikz}
\usepackage[scr=dutchcal]{mathalfa}
\usepackage{ textcomp }
\usepackage{floatflt}
\usepackage{amsthm}

\usepackage{cite}
\DeclareMathAlphabet\mathbfcal{OMS}{cmsy}{b}{n}

\theoremstyle{definition}
\newtheorem{Theorem}{Theorem}
\newtheorem{Proposition}{Proposition}

\newtheorem{Corollary}{Corollary}

\newtheorem{Remark}{Remark}
\newtheorem{Definition}{Definition}

\DeclareMathOperator*{\argmax}{arg\,max}

\begin{document}

%---------------------------------------------

%\begin{titlepage}
\title{Privacy Limits in Power-Law Bipartite Networks under Active Fingerprinting Attacks}

\author{

\IEEEauthorblockN{ Mahshad Shariatnasab, Farhad Shirani, Zahid Anwar}
\\\IEEEauthorblockA{North Dakota State University,
\\Email: $\{$mahshad.shariatnasab, f.shiranichaharsoogh,zahid.anwar$\}$@ndsu.edu,
}
}
%\and
%\IEEEauthorblockN{James Kirk\\ and Montgomery Scott}
%\IEEEauthorblockA{Starfleet Academy\\
%San Francisco, California 96678--2391\\


%Telephone: (800) 555--1212, Fax: (888) 555--1212}


\maketitle

 \begin{abstract}
This work considers the fundamental privacy limits under active fingerprinting attacks in power-law bipartite networks.
The scenario arises naturally in social network analysis, tracking user mobility in wireless networks, and forensics applications, among others. A stochastic growing network generation model --- called the popularity-based model ---  is investigated, where the bipartite network is generated iteratively, and in each iteration vertices attract new edges based on their assigned popularity values. It is shown that using the appropriate choice of initial popularity values, the node degree distribution follows a power-law distribution with arbitrary parameter $\alpha>2$, i.e. fraction of nodes with degree $d$ is proportional to $d^{-\alpha}$. An active fingerprinting deanonymization attack strategy called the augmented information threshold attack strategy (A-ITS) is proposed which uses the attacker's knowledge of the node degree distribution along with the concept of information values for deanonymization. Sufficient conditions for the success of the A-ITS, based on network parameters, are derived. It is shown through simulations that the proposed attack significantly outperforms the state-of-the-art attack strategies. 
\end{abstract}


% \fi
\IEEEpeerreviewmaketitle

\section{Introduction}

Bipartite networks model a range of application scenarios in social network analysis \cite{capocci2006preferential,newman2001clustering}, tracking mobility in wireless  networks \cite{takbiri2017limits , takbiri2019asymptotic, de2013unique,li2020drive2friends}, pandemic-related contact tracing \cite{dar2020applicability}, and security and forensics \cite{vassio2017users}. In this work, we consider bipartite networks whose vertices are partitioned into \textit{user vertices} and \textit{group vertices}, where the group vertices may represent social network groups, locations visited by users, users' online activities and browsing habits, etc.  For instance, in social networks, the users' group memberships are modeled using a bipartite network \cite{kruegel,fire2014online,su2017anonymizing}, where an edge between a user vertex and a group vertex indicates that the user is a member of that group. 

%The fundamental limits of user privacy in such networks --- set of necessary and sufficient conditions on the network statistics under which the user identities can be deanonymized by under various attack scenarios --- is of significant interest. For instance, 
 Companies use tracking tools to monitor users' online activities at varying level of intrusiveness. Sophisticated technologies such as third-party cookies, web beacons and click streams track internet addresses, order in which pages are viewed, and even the location of users when browsing websites \cite{verizon-tracking2021}. 
%In December 2021, Verizon customers discovered that they were covertly enrolled in a customer experience program \cite{verizon-tracking2021} that tracked information about browsing history, location, apps and contacts.
%In \cite{SurfShark}, it is detailed how most commonly-used websites use as many as a hundred trackers each to monitor users' activities. The data collected by such trackers can be used to construct \textit{fingerprints} for network users based on their activities. 
This data collection can be used to construct a `digital fingerprint' for network users.
In this work, we wish to find out \textit{when can user fingerprinting via data collection lead to deanonymization?}
%The resulting theoretical privacy guarantees are essential to ensure compliance with national and international privacy laws such as European Union's recently enacted  General Data Protection Regulation (GDPR) and the California Consumer Privacy Act (CCPA).
%This comes in the wake of increased concern by consumers about how companies are using their data and regulators stepping up privacy requirements. After Europe enacted the General Data Protection Regulation (GDPR) and California introduced the California Consumer Privacy Act (CCPA) several US states are following suit and companies are facing millions of dollars in fines \cite{techcrunch-grindr, BBC-HM} for noncompliance. 
In particular, we study the privacy limits in bipartite networks under active fingerprinting attacks, where, an anonymous victim is targeted by an attacker (e.g. the victim visits a malicious website), and the attacker queries her group memberships sequentially (e.g. by querying the browser history). The attacker constructs a fingerprint for the victim based on the received query responses, and by comparing this fingerprint to that of the network users, which is acquired through scanning the publicly available network graph, it identifies the victim. The problem was initially introduced and studied by  Wondracek et al. \cite{kruegel}, where an attack strategy was proposed and its effectiveness was illustrated in simulations of the attack scenario in real-world networks. The fundamental privacy limits were studied under various assumptions on the graph network in \cite{su2017anonymizing,shirani2017information,shirani2018optimal,shariatnasab2021fundamental}.

In \cite{shariatnasab2021fundamental}, we introduced a stochastic graph generation model, called the popularity-based model, proposed the information threshold strategy (ITS), and derived its fundamental performance limits in terms of expected number of queries necessary for successful denaonymization with vanishing probability of error as the graph size grows asymptotically large. The ITS strategy queries the group memberships of the victim starting with the first group in the network, and at each step, finds the \textit{information value} of each user which captures the likelihood of that user being the victim given the  query responses. It identifies a user as the victim if the information value passes a predetermined threshold. The analytical techniques in \cite{shariatnasab2021fundamental} leverage ideas from data transmission over channels with feedback \cite{burnashev1976data}. 
The ITS is agnostic to the network degree distribution. That is, it does not choose the groups to be queried based on their sizes. In this work, we improve the ITS and propose the Augmented Information Threshold Strategy (A-ITS) in which the attacker chooses which group to query based on the group sizes. The performance analysis of such strategy is challenging and requires characterizing the group degree distribution as well as the memory structure of the edges in the network. The analysis of the degree distribution and edge memory (Section \ref{sec:gen_deg_mem}) may be of independent interest in graph analysis applications as well. In summary, the contributions of this work are as follows:
\begin{itemize}[leftmargin=*]
    \item We derive the node degree distribution under the proposed popularity-based graph generation models.
    \item We show that with the appropriate choice of initial popularities, the degree distribution follows a power-law with arbitrary parameter $\alpha>2$.
    \item We show that under a sparsity condition on the number of graph edges, which requires the number of edges to grow
    linearly in the number of users, the user fingerprints are \textit{almost} memoryless.
    \item We propose the A-ITS strategy which leverages the attacker's knowledge of the node degree distribution, derive sufficient conditions for its success, and provide simulation results to illustrate the  performance gains of A-ITS as compared with ITS.
\end{itemize}

\textit{Notation:} The random variable $\mathbbm{1}_{\mathcal{E}}$ is the indicator of the event $\mathcal{E}$.
 The set $\{n,n+1,\cdots, m\}, n,m\in \mathbb{N}$ is represented by $[n,m]$, and for the interval $[1,m]$, we use the shorthand notation $[m]$. 
 For a given $n\in \mathbb{N}$, the $n$-length vector $(x_1,x_2,\hdots, x_n)$ is written as $x^n$. For $x\in \mathbb{R}$, we have defined $\lfloor x \rfloor\triangleq \max\{i\leq x| i\in \mathbb{Z}\}$ and $\lceil x \rceil\triangleq \min\{i\geq x| i\in \mathbb{Z}|\}$.

\section{Graph Generation, Degree Distribution, and Fingerprint Memory}
\label{sec:gen_deg_mem}
This section introduces the stochastic graph model, characterize the resulting degree distribution, and derives several statistical properties  which are used in the sequel. 
\subsection{Popularity-based Bipartite Graph {Generation} Model}
\label{sec:gen}

A bipartite graph is formally defined as follows.
%This is a subclass of the popularity-based growing models studied in \cite{shariatnasab2021fundamental}, which encompass the bipartite Erd\"{o}s-R\'{e}nyi \cite{erodos1959random}, bipartite stochastic block \cite{Feldman,florescu2016spectral}, and bipartite linear and sublinear  preferential attachment  \cite{peruani2007emergence,capocci2006preferential,kunegis2013preferential,borrel2005preferential} models. 

\begin{Definition}[\textbf{Bipartite Graph}]
A bipartite graph $\mathcal{G}=(\mathcal{V}_1,\mathcal{V}_2,{\mathcal{E}})$, has vertex set $\mathcal{V}_1\bigcup \mathcal{V}_2$ and edge set ${\mathcal{E}}\subseteq \{(v_{1,i},v_{2,j})|v_{1,i}\in \mathcal{V}_1, v_{2,j}\in \mathcal{V}_2\}$, where $\mathcal{V}_1\cap \mathcal{V}_2=\phi$. The sets $\mathcal{V}_1$ and $\mathcal{V}_2$ are called the left-vertices and right-vertices of $\mathcal{G}$, respectively. We define $m\triangleq |\mathcal{V}_1|$, $n\triangleq|\mathcal{V}_2|$, and  $\Delta\triangleq |\mathcal{E}|$.
\end{Definition}

\begin{Definition}[\textbf{Neighbors and Degree of a Vertex}]
For right-vertex $v_{2,j},j\in [n]$, the set $\mathcal{V}_{2,j}\triangleq \{v_{1,i}|(v_{1,i},v_{2,j})\in \mathcal{E}\}$ is called the set of neighbors of $v_{2,j}$, and $D_{j}\triangleq|\mathcal{V}_{2,j}|$ is called the degree of $v_{2,j}$. The set of neighbors and degree of left-vertices $v_{1,i},i \in [m]$ are defined similarly. 
 \end{Definition}

 The left-vertices in a bipartite graph are assigned \textit{fingerprints} based on their connections to the right-vertices. That is, the fingerprint consists of a vector of indicator variables indicating the existence or lack of existence of edges between the left-vertex under consideration and each of the right-vertices. This is formalized below.  
 \begin{Definition}[\textbf{Fingerprint of a Left-Vertex}]
For a given left-vertex $v_{1,i},i\in [m]$, the sequence of indicator variables $(\mathbbm{1}\{v_{1,i}\in \mathcal{V}_{2,1}\},\mathbbm{1}\{v_{1,i}\in \mathcal{V}_{2,2}\},\cdots, \mathbbm{1}\{v_{1,i}\in \mathcal{V}_{2,n}\})$ is called the fingerprint vector of $v_{1,i}$. 
 \end{Definition}
%For instance, in a social network, the users' group memberships can be represented by a bipartite graph $\mathcal{G}=(\mathcal{U},\mathcal{R},\mathcal{E})$, where $\mathcal{U}$ is the set of users, $\mathcal{R}$ is the set of groups, and $\mathcal{E}\subseteq \{(u,r)|u \in \mathcal{U},r\in \mathcal{R}\}$, where $(u,r)\in \mathcal{E}$ if and only if user $u$ is a member of group $r$. 

The popularity-based generation process is as follows:
\\\noindent \textbf{Initiation:} Fix the model parameter $\mu\in \mathbb{N}$, and let  $\Delta\triangleq\mu n$.  The iterative process is initiated by considering a bipartite graph $\mathcal{G}(0)=(\mathcal{V}_1, \mathcal{V}_2,\mathcal{E}(0))$, where the vertex sets  $\mathcal{V}_1\triangleq \{v_{1,1},v_{1,2},\cdots, v_{1,m}\}$, and $\mathcal{V}_2\triangleq \{v_{2,1},v_{2,2},\cdots, v_{2,n}\}$ are fixed through the iteration process, and $\mathcal{E}(0)=\phi$, i.e. there are no edges in the initial graph.  Each vertex $v_{2,j}, j\in [n]$ is assigned an initial popularity value $\tau_j(0)>0$.  The graph is generated in $\Delta$ iterative steps, where at each step, a single edge is added to $\mathcal{E}$ as described in the sequel.
\\\textbf{Step t:}  At step $t\in [\Delta]$, an edge $(v_{1,I_t}, v_{2,J_t}), I_t\in [m],J_t\in [n]$ is chosen as described next and added to the bipartite graph, i.e. $\mathcal{E}(t)=$ $\mathcal{E}(t-1)\cup \{(v_{1,I_t}, v_{2,J_t})\}$.  First, a right-vertex $v_{2,J_t}$ is chosen from the set  $\mathcal{V}_2$ by choosing $J_t$ randomly according to the probability distribution $\mathbf{P}(t)=(P_1(t),P_2(t),\cdots,$ $P_n(t))$ defined below:
\begin{align*}
    P_j(t)\triangleq \frac{\tau_{j}(t-1)}{\sum_{j'=1}^n \tau_{j'}(t-1)}, j \in [n] 
\end{align*}
% \\
% \textcolor{blue}{shouldn't j' be defined above?}
% \\
 Next, a left-vertex $v_{1,I_t}$ is chosen randomly and uniformly from the set $[m]- \mathcal{V}_{2,J_t}(t-1)$.  The edge $(v_{1,I_t},v_{2,J_t})$ is added to the edge set. 
The popularity values are updated as  $\tau_j(t)= \tau_j(t-1)+\mathbbm{1}( J_t=j),     j\in [n]$. 
%  \\ \textcolor{blue}{how was $J_t$ calculated?} \\
Equivalently, $\tau_j(t)= D_{j}(t)+\tau_j(0)$, where $D_{j}(t)$ is the degree of $v_{2,j}$ at time $t$.
 \setlength\arrayrulewidth{1pt}
 \begin{table*}
  \centering
\begin{tabular}{|ll|ll|ll|}
\hline
 $g$:& bipartite graph & $n$:  & $\#$ of groups &  $m$: & $\#$ of users    \\ 
 \hline
 $\mu$: & groups average size & $\mathcal{V}_1 :$ & users' vertex set & $\mathcal{V}_2$: & groups' vertex set   \\ \hline
$\Delta$: & $\#$ of edges & $\ell$: & $\#$ of attributes & $\tau_{j}(t)$: & popularity of $v_{2,j}$ at time $t$   \\ \hline
$D_j(t)$:& degree of $v_{2,j}$ at time t  &  $\mathcal{V}_{1,i}$:& left-vertices connected to $v_{1,i}$ &$\mathcal{V}_{2,j}$:& right-vertices connected to $v_{2,j}$  \\ 
\hline
\end{tabular}
\caption*{Notation Table: Random Bipartite Graphs}
\vspace{-.2in}
\end{table*}
 We investigate the degree distribution of bipartite graph under the following asymptotic regime:
i) number of left-vertices $m$ is taken to be asymptotically large, i.e. $m\to \infty$, ii) number of right-vertices $n$ grows linearly in $m$, i.e. $m=\beta n$ for a fixed $\beta>0$. 
iii) average value of right-vertex degrees is constant as the network grows. That is, $\Delta=  \mu n= \frac{\mu}{\beta}m, \mu \geq 1$, so that the average degree $\mu$ is constant in $n$. The latter condition is a {sparsity} condition, which is analogous to the scale-free property in linear and sublinear preferential attachment models  \cite{peruani2007emergence,capocci2006preferential,kunegis2013preferential,borrel2005preferential}. 

\subsection{Degree Distribution and the Power-law}
A critical feature of scale-free network generation models such as the well-studied Barb\'asi-Albert models, is that the resulting degree distribution follows a power-law. That is, the expected number of vertices with degree $d$ is proportional to $d^{-\alpha}$ for some $\alpha>0$. Such power-law behavior is observed empirically in real-world graphs such as social network group memberships, wireless mobility networks, and online shopping habits \cite{capocci2006preferential,newman2001clustering}. In the following, we show that the degree distribution of the right-vertices under the generation process described in Section \ref{sec:gen}, follows the power-law distribution with parameter $\alpha>0$, where $\alpha$ can be controlled by the appropriate choice of the initial popularity values. 

\begin{Definition}[\textbf{$(\mu, n,m,\alpha)$-bigraph}]
Given $\mu,\alpha>0$ ,and $n,m\in \mathbb{N}$,
a $(\mu, n,m,\alpha)$-bigraph is a bipartite graph with $n$ left-vertices, $m$ right-vertices, $\Delta=n\mu$ edges, and initial popularity values distributed according to:
\begin{align*}
    P(\tau_j(0)=k)=\frac{1}{\zeta(m,\alpha) k^\alpha}, k\in [m],
\end{align*}
where $\zeta(m,\alpha)\triangleq \sum_{i=1}^m \frac{1}{i^{\alpha}}$, and the initial popularity values are mutually independent.
\end{Definition}
Note that $\lim_{m\to \infty} \zeta(m,\alpha)$ is the well-studied Riemann-Zeta function (e.g. \cite{titchmarsh1986theory}), evaluated at $\alpha$.
%  \\ \textcolor{blue}{should $s$ be defined?} \\
The following theorem shows that  given $\alpha>2$, the degree distribution of a $(\mu, n,m,\alpha)$-bigraph converges to the power-law distribution with parameter $\alpha$ as $m\to \infty$. 

\begin{Theorem}[\textbf{Power-law in Popularity-based Models}]
\label{th:1}
Fix $\mu,\beta>0$ and $\alpha>2$, and let $\mathcal{G}_m=(\mathcal{V}_1,\mathcal{V}_2,\mathcal{E})$ be a sequence of $(\mu, n,\beta n,\alpha)$-bigraphs. Then, 
\vspace{-.1in}
\begin{align}
    P(D_{j}(\Delta)=k)=\frac{c}{k^\alpha}+o(\frac{1}{k^\alpha}), j\in [n],
    \label{eq:th:1}
\end{align}
\vspace{-.1in}
where $k= o(n^{\frac{-\beta}{\alpha}})$, and $\beta\triangleq max(-1,2-\alpha)$.
\end{Theorem}

% \textit{Proof Outline:} \textcolor{red}{The proof is provided in Appendix \ref{Ap:th:1}. We provide an outline in the following. 
% For a given vertex $v_{2,j}, j\in[n]$, we define the sequence of indicator functions $E_t\triangleq \mathbbm{1}(J_t=j), t\in[\Delta]$, and note that $D_{2,j}= \sum_{t=1}^{\Delta}E_t$. In order to have $D_{2,j}=k$, there are ${\Delta \choose k}$ possible values for the vector $(E_t=\mathbbm{1}(J_t=j), t\in[\Delta])$. We show that all of these values have probability equal to 
% \[ \sum_{i=1}^m \frac{1}{\zeta(m,\alpha)i^\alpha}
%     \times \frac{{i+k-1 \choose i-1} { \upsilon+\Delta-k-1\choose \Delta-k}}
%     {{\upsilon+i+\Delta-1 \choose \Delta}
%   },\]
% where $\upsilon$ is the sum of all initial popularity values, and approaches $\frac{n \zeta(m,\alpha-1)}{\zeta(m,\alpha)}$ with probability one as $n\to\infty$ by Proposition \ref{Prop:1} below. The rest of the proof follows by  using the following refinement of Stirling's approximation ${n\choose k}= c_k \sqrt{\frac{n}{2\pi k (n-k)}}e^{nh_b(\frac{k}{n})}, n\in \mathbb{N}, k\leq n$, where $c_k\in [e^{\frac{-1}{6k}},1]$, and $h_b(p)=-p\ln(p)-(1-p)\ln(1-p)$ is the binary entropy function measured in nats (e.g. \cite{stanica2001good})}
Please refer to Appendix \ref{Ap:th:1}.

%\textit{Proof Outline:}  We provide an outline in the following. 
%For a given vertex $v_{2,j}, j\in[n]$, we define the sequence of indicator functions $E_t\triangleq \mathbbm{1}(J_t=j), t\in[\Delta]$, and note that $D_{2,j}= \sum_{t=1}^{\Delta}E_t$. In order to have $D_{2,j}=k$, there are ${\Delta \choose k}$ possible values for the vector $(E_t=\mathbbm{1}(J_t=j), t\in[\Delta])$. We show that all of these values have probability equal to 
%\[ \sum_{i=1}^m \frac{1}{\zeta(m,\alpha)i^\alpha}
%    \times \frac{{i+k-1 \choose i-1} { \upsilon+\Delta-k-1\choose \Delta-k}}
%    {{\upsilon+i+\Delta-1 \choose \Delta}
%   },\]
%where $\upsilon$ is the sum of all initial popularity values, and approaches $\frac{n \zeta(m,\alpha-1)}{\zeta(m,\alpha)}$ with probability one as $n\to\infty$ by Proposition \ref{Prop:1} below. The rest of the proof follows by  using the following refinement of Stirling's approximation ${n\choose k}= c_k \sqrt{\frac{n}{2\pi k (n-k)}}e^{nh_b(\frac{k}{n})}, n\in \mathbb{N}, k\leq n$, where $c_k\in [e^{\frac{-1}{6k}},1]$, and $h_b(p)=-p\ln(p)-(1-p)\ln(1-p)$ is the binary entropy function measured in nats (e.g. \cite{stanica2001good}) 
\begin{Proposition}[\textbf{Concentration of Initial Popularities}]
\label{Prop:1}
Fix $\alpha>2$, and let $Y\triangleq \sum_{j=1}^n \tau_j(0)$ be the total initial popularities of the right-vertices. Then,
\begin{align}
&\mathbb{E}(Y)=  \frac{n \zeta(m,\alpha-1)}{\zeta(m,\alpha)},\label{eq:P2:1}
\\& Var(Y)\leq
\frac{n\zeta(m,\alpha-2)}{\zeta(m,\alpha)},\label{eq:P2:2}
\\&    P(|Y-\mathbb{E}(Y)|>\epsilon \mathbb{E}(Y))=o(n^{\beta}),\label{eq:P2:3}
\end{align}
where $\beta= max(-1,2-\alpha)$.
\end{Proposition}
\begin{proof}
Please refer to Appendix \ref{Ap:Prop:1}. 
\end{proof}

% \begin{Remark}
% It can be observed from the proof of Theorem \ref{th:1} that the degree distribution depends on the total number of edges $\Delta$ through the model parameter $\mu$ which determines the constant $c$ in Equation \eqref{eq:th:1}.
% \end{Remark}
It should be noted that the proposed popularity-based generation models can be used to generate degree distributions other than power-law distributions by appropriately choosing the parameter $\alpha$. For instance, the following proposition shows that the degree distribution converges to a geometric distribution as $m\to \infty$ when $\alpha\to \infty$, i.e. when all initial popularity values are set to be equal to one. 
\begin{Remark}
It can be observed from the proof of Theorem \ref{th:1} that the degree distribution depends on the total number of edges $\Delta$ through the model parameter $\mu$ which determines the constant $c$ in Equation \eqref{eq:th:1}.
\end{Remark}
\begin{Proposition}[\textbf{Geometric Degree Distribution}]
\label{Prop:2}
Fix $\mu,\beta>0$, and let $\mathcal{G}_m=(\mathcal{V}_1,\mathcal{V}_2,\mathcal{E})$ be a sequence of $(\mu, n,\beta n,\alpha)$-bigraphs. Then, 
\begin{align*}
    \lim_{n\to \infty }\lim_{\alpha\to \infty} P(D_j(\Delta)=k)=\left(\frac{\mu}{1+\mu}\right)^k \frac{1}{1+\mu}, j\in [n].
\end{align*}
\end{Proposition}
\begin{proof}
Please refer to Appendix \ref{Ap:Prop:2}.
\end{proof}

\subsection{Asymptotically Memoryless Fingerprints}

 A major obstacle in analyzing the asymptotic properties of bipartite networks and performance limits of attack algorithms is the \emph{memory structure} in the edges connecting a given left-vertex to the right-vertices.  That is, the generation model induces correlation among the edges, and this prohibits the conventional large deviations techniques  which have been used in deriving theoretical performance limits in similar scenarios in group testing \cite{naghshvar2013active} and communications \cite{burnashev1976data} problems. In \cite{shariatnasab2021fundamental}
 we have shown that if all initial popularity values are equal to one (i.e. $\alpha\to \infty$), the memory in the left-vertices' fingerprints is \textit{weak}, so that the joint distribution of a given fingerprint approaches a product distribution. 
 In this section, we extend the result to  $\alpha>2$, i.e. power-law bigraphs. 
 %The following generalizes Proposition 1 in \cite{shariatnasab2021fundamental} which was proven for $\alpha\to \infty$, and extends the result to $\alpha>2$. 
  \begin{Proposition}[\textbf{Group Size Correlation}]
 \label{Prop:3}
Let $\alpha>2, \mu>1$ and $\beta>0$ .  For a  $(\mu, n,\beta n,\alpha)$-bigraph, the following holds:
\begin{align}
& \label{eq:prop1:1}\mathbb{E}(D_{j}(\Delta))=\mu, j\in [n],
\\\label{eq:prop1:2}
    & \mathbb{E}(D^2_{j}(\Delta))= O(1), j\in [n],
\\&\label{eq:prop1:3} \mathbb{E}(D_{i}(\Delta)D_{j}(\Delta))= \mu^2+O(\frac{1}{n}), i\neq j,\\
&\label{eq:prop1:2.5} \mathbb{E}(D_{1}(\Delta)D_{2}(\Delta)\cdots D_{\xi}(\Delta))= \mu^\xi(1+\xi O(\frac{1}{n})), \xi\in [n],
    \\\label{eq:prop1:4}
   & \mathbb{E}(D^2_{j}(\Delta)D_{2}(\Delta)D_{3}(\Delta)\cdots D_{\xi}(\Delta))\leq
    \mu^{\xi-1} \mathbb{E}(D^2_{1}(\Delta))
   % \mathbb{E}(D_{\Delta,2}D_{\Delta,3}\cdots D_{\Delta,n'-1})
   , \xi\in [n],
    \\&\label{eq:prop1:5}
    \mathbb{E}(D_{1}(\Delta)D_{2}(\Delta)D_{3}(\Delta)\cdots D_{\xi}(\Delta))\leq \mu^{\xi}, \xi\in [n].
\end{align}
 \end{Proposition}
 
 \begin{proof}
 Please refer to Appendix \ref{App:Prop:3}.
 \end{proof}

% \\ \textcolor{blue}{in (8) is it $D_1$ or $D_i$?} \\
Following the arguments in  \cite{shariatnasab2021fundamental}, for a given bigraph satisfying the conditions in  Proposition \ref{Prop:3}, the left-vertex fingerprints are `almost' memoryless as stated below. 
\begin{Proposition}[\textbf{Memoryless Fingerprints}]
\label{Prop:4}
Let $\alpha>2, \mu>1$ and $\beta>0$ .  For a  $(\mu, n,\beta n,\alpha)$-bigraph. Consider the partial fingerprint $\mathbf{R}\triangleq (R_{i,j_k})_{k\in [n']}, j_k\in [n], n' \in [n]$ of left-vertex $v_{1,i}, i\in [m]$. The following holds:
\begin{align*}
&(1-\frac{n'\lambda(m,\alpha)}{m}) \prod_{k=1}^{n'}P_{R}(s_k)
\leq 
     P_{\mathbf{R}}(s^{n'})\leq  
 e^{\frac{\lambda(m,\alpha)}{\beta}}\prod_{k=1}^{n'}P_{R}(s_k),
 \end{align*}
for all $s^{n'}\in \{0,1\}^{n'}$, where $\lambda(m,\alpha)\triangleq \mu+\zeta(m,\alpha-1)$ and  
 $P_R(\cdot)= P_{R_{i,j}}(\cdot), i\in [m], j\in [n]$. Furthermore, assume that $n'>\frac{m}{\mu}$ and $\sum_{i=1}^{n'}\mathbbm{1}(s_i=1)=o(n)$ for some constant finite number $C>0$.  Then, there exists $c'>0$ whose value only depends on $\mu$ and $\beta$ such that:
 \begin{align*}
   & c'\prod_{k=1}^{n'}P_{R}(s_k)(1+o(1))\leq P_{\mathbf{R}}(s^{n'})
 \leq \prod_{k=1}^{n'}P_{R}(s_k)(1+o(1)),
 \end{align*}
 as $n \to \infty$, where $s^{n'}\in \{0,1\}^{n'}$.  
\end{Proposition}
The following sparsity result holds for the fingerprint vector of the left-vertices.
\begin{Proposition}[\textbf{Sparsity of the Fingerprint Vector}]
\label{Prop:5}
 Let $\alpha>2, \mu>1$ and $\beta>0$ .  For a  $(\mu, n,\beta n,\alpha)$-bigraph. Consider the partial fingerprint $\mathbf{R}\triangleq (R_{i,j_k})_{k\in [n']}, j_k\in [n], n' \in [n]$ of left-vertex $v_{1,i}, i\in [m]$, there exists a constant $c>0$ such that:
\begin{align}
   & P(C_i\geq \ell)\leq c2^{-nD_b(\frac{\lambda(m,\alpha)}{m}(1+\psi)||\frac{\lambda(m,\alpha)}{m})},
\end{align}
where $\ell=\frac{1}{\beta}(\lambda(m,\alpha))(1+\psi)$, $\psi \in (0, \frac{m}{\mu+\zeta(m,\alpha-1)}-1)$, and $D_b(p||q)=p\log{\frac{p}{q}}+ (1-p)\log{\frac{1-p}{1-q}}$ is the binary Kullback-Leibler divergence. Particularly, let $\psi_n>0, n\in \mathbb{N}$ such that  $\psi_n=\omega(1)$. Then, 
\begin{align}
\label{eq:prop2}
       & P(C_i\geq \psi_n)\to 0,\text{ as }n\to\infty.
\end{align}
\end{Proposition}
\begin{proof}
Please refer to Appendix \ref{App:Prop:5}.
\end{proof}
\section{Attack Strategy and Fundamental Performance Limits}
In this section, we apply the derivations in Section \ref{sec:gen_deg_mem} to investigate the fundamental privacy limits in bipartite networks under active fingerprinting attacks. 
\subsection{Attack Scenario} 
The scenario is captured by the following:
\\\noindent\textbf{Ground-Truth:} We consider the \textit{ground-truth} bipartite graph  $\mathcal{G}_0=(\mathcal{U},\mathcal{R},\mathcal{E})$ capturing users' group memberships in a given network, where i) the set of left-vertices $\mathcal{U}$ represents the set of users in the network, ii) the set of right-vertices $\mathcal{R}$ represent the set of groups in the network, e.g. social network groups, locations visited by users, users' online activities and browsing habits, etc., and iii) An edge $(u_i,r_j)\in \mathcal{E}, i\in [m], j\in [n]$ between a user $u_i$ and a group $r_j$ indicates the user's membership in the group. The ground-truth is modeled by a $(\mu,n,\beta n, \alpha)$-bigraph, where $n\in \mathbb{N}$ is the number of groups, $m=\beta n$ is the number of users, $\Delta=\mu n$ is the number of edges, and $\alpha>2$ is a network parameter which depends on the network's power-law distribution \cite{kunegis2013preferential}.

\noindent \textbf{Scanned Graph:} Prior to the start of the attack, the attacker scans the ground-truth and acquires an observation captured by the \textit{scanned graph} $\mathcal{G}_{s}$. In this work, for brevity, we assume that the scanning operation is noiseless, i.e. $\mathcal{G}_s=\mathcal{G}_0$. However, in general, the operation may be noisy, and the scanning noise depends on the users' individual privacy preferences, e.g. social network privacy settings.  The derivations provided in the sequel may be potentially extended to noisy scanned graphs using  techniques similar to the ones in \cite{shariatnasab2021fundamental} which investigated the case when $\alpha\to \infty$.

\noindent\textbf{Victim:} A \textit{victim} $u_M$ is the  user which is targeted by the
 attacker. For instance, the victim may visit a malicious website, where the attacker uses browser history sniffing techniques to 
sequentially query its group memberships \cite{wondracek2010practical}. We assume that the victim's index is chosen from the set $[m]$ randomly according to $P_M$. The distribution $P_M$ may not be uniform as users are not equally likely to fall victim to an attack, with more risk-averse users less likely to be victims in an attack. 

\noindent\textbf{Query Responses:} The attack initiates with the attacker sequentially querying the victim's  group memberships.
Generally, the \textit{query responses} are noisy, e.g. browser history sniffing techniques are imperfect and only provide noisy observations of the victim's browsing history. The noise statistics are determined by the users' software (e.g. browser \cite{smith2018browser}) and hardware specifications (e.g. CPU and memory specifications  \cite{gulmezoglu2017perfweb}),  and depend on the type of history sniffing attack.  This dependency is captured by the parameter $\theta(M)$, where $\theta:[m]\to \Theta$, and $\Theta$ is a finite set; so that the response $Y\in \{0,1\}$ to the query regarding the victim's membership in group $r_j, j\in [m]$, characterized by the indicator variable  $R_j\in \{0,1\}$,   is produced conditionally with distribution $P^{\theta(M)}_{Y|R_j}$. The following definition formalizes the stochastic model for the query responses.  
\begin{Definition}[\textbf{Noisy Query Responses}]
\label{def:QR}
Let $\ell\in \mathbb{N}$ and $P^{\theta}_{Y|R}, \theta\in \Theta$ be a collection of distributions, where $Y$ and $R$ are binary and $\Theta$ is a finite set. For sequence $j_1,j_2,\cdots,j_{\ell}\in [n]$, assume that victim's fingerprint is $(R_{j_1},R_{j_2},\cdots, R_{j_\ell})$ and received query responses are $Y_1,Y_2,\cdots,Y_\ell$. Then, 
\begin{align*}
P(Y^\ell=y^\ell|   (R_{j_i})_{i\in [\ell]}=r^\ell) =\prod_{i=1}^\ell P^{\theta(M)}_{Y|R}(y_i|r_i), y^{\ell}, r^{\ell}\in \{0,1\}^{\ell},
\end{align*}
where the parameter $\theta(M)$ takes values from $\Theta$ and its value depends on the victim's index $M$.
\end{Definition}
The attacker has access to $\theta(M)$ since it  can query the victim's hardware and software specifications. So, it can find $P^{\theta(M)}_{Y|E_0}$ and use it in calculating the users' information values as in the previous scenario.
The attacker's objective is to deanonymize the user by comparing the user fingerprints in the scanned graph $\mathcal{G}_s$ and the vector of query responses $Y^{\ell}$. An efficient attack strategy minimizes the expected number of queries necessary for successful deanonymization, e.g. minimizes $\ell$ with vanishing probability of identification error. 

\begin{Definition}[\textbf{Attack Strategy}]
 Consider an attack scenario parametrized by $(n,\beta, \mu, \alpha, \Theta,  P_M,  \\(P^{\theta(k)}_{Y|R})_{k\in [m], \theta\in \Theta})$. An attack strategy consists of a sequence of query functions $x_t: \{0,1\}^{m\times n}\times\{0,1\}^{(t-1)}\to\mathcal{R}, t\in \mathbb{N}$ and  identification functions $Id_t: \{0,1\}^{m\times n} \times \{0,1\}^t \to \mathcal{U} \cup \{e\}$, where $x_t(\mathcal{G}_s, Y^{t-1})$ outputs the group whose edge connection with the victim is queried at time $t$, and $Id_t(\mathcal{G}_s,Y^t)$  either outputs the victim's identity among the user set $\mathcal{U}$ or outputs `$e$', indicating that a unique victim has not been identified yet, 
 %\textcolor{blue}{does $e$ stand for error?}
 in which case further queries are made and the attack continues.
% \\\textbf{Without Feedback:} An attack strategy without feedback consists of a sequence of query functions $x_t:\{0,1\}^{m\times n}\to \mathcal{R}^t, t\in \mathbb{N}$ and identification functions $I_t: \{0,1\}^{m\times n} \times \{0,1\}^t \to \mathcal{U}, t\in \mathbb{N}$. The value of $t$ is fixed prior to the start of the attack, and sends a batch of $t$ queries $x_t(\mathcal{G}_s)$ are sent simultaneously. The identification function $I_t(\mathcal{G}_s,Y^t)$ outputs the attacker's best estimate of the victim's identity in $\mathcal{U}$ using $\mathcal{G}_s$ and $Y^t$.
Let $Q= min\{t\in \mathbb{N}:Id_t(\mathcal{G}_s,Y^t)\in \mathcal{U}\}$. Then, the
%(worst case) 
probability of error $P_e$ and  expected number of queries $\overline{Q}$ are defined as:
\begin{align*}
    &P_e((x_t,Id_t)_{t\in \mathbb{N}})\triangleq %\max_{\gamma(\cdot), \theta(\cdot)} 
    P(Id_Q( \mathcal{G}_s,Y^Q)\neq u_M)\\
    &\overline{Q}((x_t,Id_t)_{t\in \mathbb{N}})\triangleq %\max_{\gamma(\cdot), \theta(\cdot)} 
    \mathbb{E}(  Q), 
\end{align*}
 where
 %the maximum is taken over all $\gamma(\cdot)\in \Gamma, \theta(\cdot)\in \Theta, $ and 
 the probabilities are with respect to $M, \mathcal{G}_0, \mathcal{G}_s$ and $Y_t, t\in [Q]$. The minimum expected number of queries is defined as:
 \[{Q}^*_{\epsilon}\triangleq\inf_{(x_t,Id_t)_{t\in \mathbb{N}}}\{\overline{Q}((x_t,Id_t)_{t\in \mathbb{N}})| P_e((x_t,Id_t)_{t\in \mathbb{N}})\leq \epsilon\}.\]
\end{Definition}

\subsection{Popularity-Based Information Threshold Attack Strategy}
This section provides an attack strategy which improves upon the information threshold strategies (ITS) investigated in \cite{shirani2018optimal,shariatnasab2021fundamental}, and uses the derivations in Section \ref{sec:gen_deg_mem} to derive its fundamental performance limits in terms of minimum expected number of queries and probability of error.

\noindent\textbf{Attack Strategy:} The attacker queries the group memberships of the victim starting from the largest group $r_{j_1}, i.e., j_1= \argmax \Big\{|D_{j}|\Big|j\in [n]\Big\}$ and continues by querying the next largest group, so that $x_s=r_{j_s}$, where $r_{j_s}$ is the $s$th 
%\textcolor{blue}{$t^{th}$}
largest group. The queries continue until a particular stopping criterion is met. To explain the stopping criterion, let us define the \textit{information value} $I_{k}(t), k\in [m], t\in [n]$ of user $u_k$ and time $t$ as follows:
\begin{align}
    &\label{eq:ISN:1}I_0(k)\triangleq \log{P_M(k)}, k\in [m],\\
    % &I_t(k)= \log{\frac{\prod_{i=1}^tP^1_{E_s|E_0}(y_i|f_{k,i})}{P_{Y^t}(y^t)}}+I_0(k), k\in [m], t\in \mathbb{N}
    &\label{eq:ISN:2}I_t(k)\triangleq 
\sum_{i=1}^t    \log{\frac{P^{\theta(M)}_{Y|E_0}(y_i|f_{k,i})}{P_{Y_t}({y_i})}}+I_0(k), k\in [m], t\in [n],
\end{align}
where $P_{Y_t}(1)\triangleq \frac{1}{D_0(r_t)}$.  Intuitively, the information value $I_t(k)$ captures the attacker's belief at time $t\in [n]$ about the possibility of user $u_k, k\in [m]$ being the victim, based on the received query responses, where a large positive $I_t(k)$ indicates a strong belief that the user is the victim, and a large negative $I_t(k)$ indicates a strong belief that the user is not the victim.
The identification function $Id_t$ first determines whether the maximum information value of all users exceeds $\log{\frac{1}{\epsilon}}$, where the parameter $\epsilon>0$ affects the resulting probability of error. If there exists a user whose information value exceeds  $\log{\frac{1}{\epsilon}}$, {then} that user is identified as the victim. Otherwise, the next query is made. So,
\begin{align}
    &x_t(\mathcal{G}_s,Y^{t-1})=r_{j_t}, t\in [n]\\
&    Id_t(\mathcal{G}_s,Y^t)=  
    \begin{cases}
    u_k\qquad& \text{ if } \exists ! k\in [m]: I_t(k)>log{\frac{1}{\epsilon}}\\
    e& \text{Otherwise}
    \end{cases}, t\in [n] \label{eq:id}
\end{align}

We call this attack the Augmented-ITS (A-ITS) since it uses both information thresholds and the group degree distribution. 

\begin{Theorem}
\label{th:2}
Consider the A\text{-}ITS described above with parameter $\epsilon>0$. Let $\overline{Q}_{A\text{-}ITS}$ be the resulting expected number of queries and $P_{e,A\text{-}ITS}$ the resulting probability of error. Then,
        \begin{align*}
        &\overline{Q}_{A\text{-}ITS}\leq
        \sum_{\theta\in \Theta}P_{\Theta}(\theta)\sum_{d\geq d_{\theta}^*}\frac{n}{\zeta(m,\alpha) d^\alpha}+ i_{\theta}^*
        ,\qquad P_{e,A\text{-}ITS}\leq \frac{\epsilon}{c'},
    \end{align*}
    where
    \begin{align}
       &d_{\theta}^*\triangleq\max_{d\in [m]}\bigg\{d:\psi \leq c' \sum_{\theta\in \Theta}P_{\Theta}(\theta) \sum_{d'\geq d-1}\frac{n}{\zeta(m,\alpha) d'^\alpha}I_{d',\theta}(Y;E_0)\bigg\}, \\
       &i^*\triangleq \min_{i\in [\frac{n}{\zeta(m,\alpha) (d^*_\theta-1)^{\alpha}}]}\bigg\{i: 
      \psi \leq c' \sum_{\theta\in \Theta}P_{\Theta}(\theta) \sum_{d'\geq d_{\theta}^*}\frac{n}{\zeta(m,\alpha) d'^\alpha}\times
      \\&\nonumber\qquad\qquad  \qquad \qquad  \qquad  
      I_{d,\theta}(Y:E_0)+iI_{d_{\theta}^*-1}(Y;E_0) \bigg\},
      \\&\psi\triangleq H(M)+\log{\frac{1}{\epsilon}+i_{\max}},
    \end{align}
    where $c'$ is from Proposition \ref{Prop:4}, %\textcolor{blue}{I think $c'$ was introduced in Proposition 4} 
    the variable $N_d, d\in [m]$ is the number of groups with degree equal to $d$ in the ground-truth graph, the mutual information $I_{d,\theta}(Y;E_0)$ is evaluated with respect to $P^{d,\theta}_{Y,E_0}=P^d_{E_0}P^{\theta}_{Y|E_0}$, the variable $E_0$ is Bernoulli with parameter $\frac{d}{m}$,  $P^{\theta}_{Y|E_0}$ is the query noise with parameter $\theta$ given in Definition \ref{def:QR},  $i_{max}\triangleq \max_{y,r,\theta\in \{0,1\}\times \Theta}\log{\frac{P^{\theta}_{Y|E_0}(y|r)}{P^{\theta}_{Y}(y)}}$,
and $P_{\Theta}(\theta)\triangleq \frac{|\{u_k| \theta(k)=\theta,\}|}{m}, \theta\in \Theta$.
\end{Theorem}
\begin{proof}
Please refer to \ref{App:th:2}.
%Appendix \ref{App:th:2}.
\end{proof}
The following is a direct consequence of Theorem \ref{th:2}, and the fact that $\int_d^{\infty}x^{-\alpha}dx\leq\sum_{i=d}^\infty\frac{1}{d^\alpha}\leq\int_{d-1}^{\infty}x^{-\alpha}dx$ for $\alpha>2$:
\begin{Corollary}
Consider the A\text{-}ITS with parameter $\epsilon>0$. Assume that $|\Theta|=1$, and $P_{Y|E_0}$ is a binary symmetric distribution  with crossover probability $n_q\leq \frac{1}{2}$.  Then,
    \begin{align}
        &\overline{Q}_{ITS}\leq
        \frac{cn}{(\alpha-1)\zeta(m,\alpha) {(d^*-2)}^{\alpha-1}},\\
        &P_{e,ITS}\leq \frac{\epsilon}{c'},
        \label{eq:th3:2}
    \end{align}
    where $c$ and $c'$  are from Propositions \ref{th:1} and \ref{Prop:4}, respectively, $d^*=\max_{d\in \{3,4,\cdots,m\}}\{d|\psi\leq \frac{c'cn}{(\alpha-1)\zeta(m,\alpha)d^{\alpha-1}}$ $\times(h_b(\frac{d}{m}\ast n_q)-h_b(n_q))\}$,
    $h_b(\cdot)$ is the binary entropy function, and  $a\ast b= a(b-1)+b(a-1), a,b\in [0,1]$. 
\end{Corollary}

\section{Simulation Results}
\begin{figure}[t]
 \centering \includegraphics[width=0.8\linewidth, draft=false]{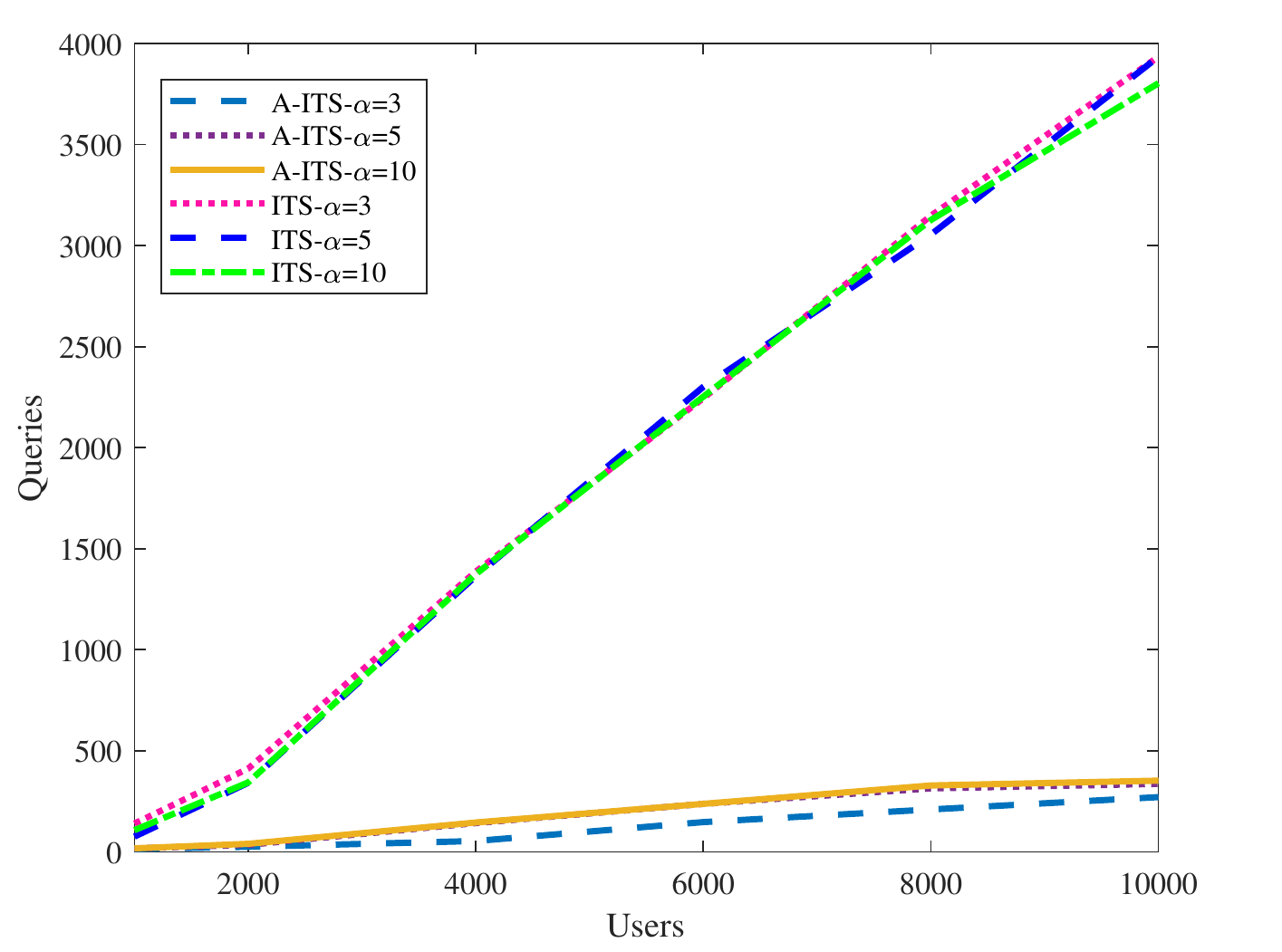}
 \caption{ Expected number of queries $\overline{Q}$ in A-ITS and ITS with error probability less than $0.05$.}
 \label{fig:sim:1}
\end{figure}
In this section, we provide a simulation of active fingerprinting attacks on synthesized bigraphs. 
In order to provide a baseline for comparison, we also investigate the performance of a natural extension of the ITS considered in \cite{shariatnasab2021fundamental}.
We generate the ground-truth with $\alpha\in\{3,5,10\}$. Furthermore, we consider $|\Theta|=1$ and a single $P_{Y|E_0}$ which is a binary symmetric distribution  with crossover probability $n_q=0.05$. We take the victim to be chosen equally likely among the users. We have simulated the attack with parameters $\mu=100$, $\epsilon=0.01$, and $\beta=0.1$ and $m= \{1000,2000,4000,6000,8000,10000\}$. For each set of parameters,  we have simulated the attack $500$ times, by generating the ground-truth five times and choosing a victim randomly and uniformly for each generation $100$ times. %\textcolor{blue}{should group size be mentioned here?} 
Figure \ref{fig:sim:1} shows the performance of A-ITS and ITS in terms of expected number of queries. We have chosen the parameter $\epsilon$ such that the empirical observed probability of error is close to $0.05$ for each set of simulation parameters. It can be observed  A-ITS  significantly outperforms the ITS. This suggests that the attacker can make significant improvements by leveraging its  knowledge of the group sizes in its  choice of queries. The expected number of queries is increasing in $\alpha$, and it grows linearly in the number of users $m$. The latter is in agreement with the observations made in \cite{shariatnasab2021fundamental} for the $\alpha\to \infty$ scenario.
%\textcolor{blue}{Maybe discuss here that alpha doesn't make a noticeable difference for small values of users? Also perhaps make a comparison with real world users and groups. The number of queries seems to be increasing linearly.}

\section{Conclusion}
The fundamental privacy limits under active fingerprinting attacks in power-law bipartite networks was considered.
The popularity-based model was investigated, and it was shown that using the appropriate choice of initial popularity values, its node degree distribution follows a power-law distribution with arbitrary parameter $\alpha>2$. An active fingerprinting deanonymization attack strategy called the augmented information threshold attack strategy (A-ITS) was proposed, and sufficient conditions for its success, based on network parameters, were derived. It was shown through simulations that the proposed attack significantly outperforms the state-of-the-art attack strategies.
\newpage

\begin{appendices}
\section{Proof of Theorem \ref{th:1}}
\label{Ap:th:1}
Fix $\epsilon>0$. Let $Y_j= \sum_{j'\neq j}\tau_{j'}(0)$ be the sum of all initial popularity values except the initial popularity value of the $j$th right-vertex. Note that $\mathbb{E}(Y_j)= \frac{m-1}{m}\mathbb{E}(Y)= \frac{(m-1)n\zeta(m,\alpha-1)}{m\zeta(m,\alpha)}$ where the last equality is due to Proposition \ref{Prop:1}.
We let $\mathcal{A}$ be the event that $|\Theta_j-\mathbb{E}(Y_j)|>\epsilon \mathbb{E}(Y_j)$, and write
\begin{align*}
    &P(D_j(\Delta)=k)= P(D_j(\Delta)=k,\mathcal{A}^c)+ P(D_j(\Delta)=k, \mathcal{A})
    \end{align*}
Note that from the arguments in the proof of Proposition \ref{Prop:1} we have $P(\mathcal{A})= o(n^{\beta})$ and from the theorem statement, we have $k=o(n^{\frac{-\beta}{\alpha}})$. Consequently, $P(\mathcal{A})=o(k^{-\alpha})$ as shown below:

\begin{align*}
P(\mathcal{A}) k^{\alpha}= \frac{P(\mathcal{A})}{n^{\beta}} k^{\alpha} n^{\beta}=
\frac{P(\mathcal{A})}{n^{\beta}} (k n^{\frac{\beta}{\alpha}})^{\alpha} \to 0 \text{ as } n\to \infty.
\end{align*}

So,
\begin{align*}
P(D_j(\Delta)=k)=    P(D_j(\Delta)=k,\mathcal{A}^c)+ o(k^{-\alpha}).
\end{align*}
Next, we investigate $P(D_j(\Delta)=k,\mathcal{A}^c)$. Note that $\Theta_j$ and $\tau_j(0)$ are independent variables. We have:
\begin{align*}
&
P(D_j(\Delta)=k,\mathcal{A}^c)= 
\sum_{\upsilon\in \mathcal{T}}\sum_{i=1}^m P(Y_j=\upsilon)P( \tau_j(0)=i)
    P(D_j(\Delta)=k|\tau_j(0)=i,Y=\upsilon+i)
    \\& \leq \max_{\upsilon\in \mathcal{T}}\sum_{i=1}^m P(\tau_j(0)=i)
  P(D_j(\Delta)=k|\tau_j(0)=i,Y=\upsilon+i),
  \end{align*}
  where $\mathcal{T}\triangleq[\lfloor\mathbb{E}(Y_j)(1-\epsilon)\rfloor,\lceil\mathbb{E}(Y_j)(1+\epsilon)\rceil]$, and in the last inequality we have used the fact that the maximum is greater that the average. Similarly, 
  \begin{align*}
& P(D_j(\Delta)=k,\mathcal{A}^c){\geq}  
P(\mathcal{A}^c) \sum_{i=1}^mP( \tau_j(0)=i)  \min_{\upsilon\in \mathcal{T}} P(D_j(\Delta)=k|\tau_j(0)=i,Y=\upsilon+i) 
\\&
 = (1-o(k^{-\alpha}))\min_{\upsilon\in \mathcal{T}}\sum_{i=1}^m P(\tau_j(0)=i)
 P(D_j(\Delta)=k|\tau_j(0)=i,Y=\upsilon+i),
  \end{align*}
  where we have used the fact that $P(\mathcal{A}^c) =  \sum_{\upsilon\in \mathcal{T}} P(Y_j=\upsilon) =1- o(k^{-\alpha})$. 
Furthermore, 
\begin{align*}
 &  P(D_j(\Delta)=k|\tau_j(0)=i,Y=\upsilon+i)
  =
  \sum_{b^{\Delta}\in \mathcal{B}_{\Delta,k}}\prod_{t=1}^{\Delta}  P(\tau_j(t)=x_t|\tau_j(t-1)=x_{t-1},Y=\upsilon+i),
\end{align*}
where $x_t\triangleq \sum_{\ell=1}^{t} b_\ell+\tau_j(0), t\in [\Delta]$, $x_0\triangleq \tau_j(0)$, and $\mathcal{B}_{\Delta,\ell}\triangleq \{b^{\Delta} \in \{0,1\}^{\Delta}| \sum_{\ell=1}^\Delta b_\ell=k\}$. Note that 
\begin{align*}
& P(\tau_j(t)=x_{t-1}+1|\tau_j(t-1)=x_{t-1},Y=\upsilon+i)=
 \frac{x_{t-1}}{\upsilon+i+ t-1}, 
 \\& P(\tau_j(t)=x_{t-1}|\tau_j(t-1)=x_{t-1},Y=\upsilon+i)=
 1- \frac{x_{t-1}}{\upsilon+ i+t-1}, 
\end{align*}
Next, we argue $\prod_{t=1}^{\Delta}  P(\tau_j(t)=x_t|\tau_j(t-1)=x_{t-1},Y=\upsilon+i)$ is equal for all  $b^{\Delta}\in \mathcal{B}_{\Delta,k}$. 
Note that any given pair  $b^{\Delta}, b^{'\Delta}\in \mathcal{B}_{\Delta,k}$ are permutations of each other since $\sum_{\ell=1}^{\Delta}b_\ell = \sum_{\ell=1}^{\Delta}b'_\ell$. Let $b'_i= b_{\sigma(i)}, i\in \Delta$, where $\sigma\in S_n$ and $S_n$ is the symmetric group of permutations over $[\Delta]$. Furthermore, define $ x'_t\triangleq \sum_{\ell=1}^{t} b'_\ell+\tau_j(0), t\in [\Delta]$.
As a first step, we only consider transpositions. To elaborate, we show that 
\begin{align}
 &  \prod_{t=1}^{\Delta}  P(\tau_j(t)=x_t|\tau_j(t-1)=x_{t-1},Y=\upsilon+i)=
\prod_{t=1}^{\Delta}  P(\tau_j(t)=x'_t|\tau_j(t-1)=x'_{t-1},Y=\upsilon+i),
\label{eq:Ap:Prop1:1}
\end{align}
where $\sigma'= (\kappa,\kappa+1)$ so that $\sigma'$ is the transposition which swaps $\kappa$ with $\kappa+1$ for a given $\kappa \in [\Delta-1]$. Note that if $b_\kappa=b_{\kappa+1}$, then the proof is trivial since $b^{\Delta}=b^{'\Delta}$. Assume that $b_\kappa=0, b_{\kappa+1}= 1$. 
Note that in this case
\begin{align*}
  & P(\tau_j(t)=x_t|\tau_j(t-1)=x_{t-1},Y=\upsilon+i)=
    P(\tau_j(t)=x'_t|\tau_j(t-1)=x'_{t-1},Y=\upsilon+i), t\neq \kappa,\kappa+1,
    \\
    &
 P(\tau_j(\kappa)=x_{\kappa}|\tau_j(\kappa-1)=x_{\kappa-1},Y=\upsilon+i)=1-\frac{x_{\kappa-1}}{\upsilon+i+\kappa-1},
     \\
    &
 P(\tau_j(\kappa+1)=x_{\kappa+1}|\tau_j(\kappa)=x_{\kappa},Y=\upsilon+i)=\frac{x_{\kappa-1}}{\upsilon+i+\kappa},
    \\
    &
 P(\tau_j(\kappa)\!=\!x'_{\kappa}|\tau_j(\kappa-1)=x'_{\kappa-1},Y=\upsilon\!+\!i)\!= \frac{x_{\kappa-1}}{\upsilon+i+\kappa-1},
     \\
    &
 P(\tau_j(\kappa+1)=x'_{\kappa+1}|\tau_j(\kappa)=x'_{\kappa},\Theta=\upsilon+i)=1-\frac{x_{\kappa-1}+1}{\theta+i+\kappa},
\end{align*}
which proves Equation \eqref{eq:Ap:Prop1:1}. The proof for the case where $b_{\kappa}=1, b_{\kappa+1}=0$ follows similarly. Next, we extend the argument to arbitrary $\sigma\in S_n$. It is well-known that a decomposition $\sigma= \sigma_{1}\circ \sigma_{2}\circ \cdots \circ \sigma_{r}, r\in \mathbb{N}$ always exists, where  $\sigma_{j}=(\ell_{j}, \ell_{j}+1), \ell_{j}\in [\Delta-1]$ is the transposition which swaps $\ell_{j}$ and $\ell_{j}+1$
(e.g. \cite{isaacs2009algebra}). 
The proof of Equation \eqref{eq:Ap:Prop1:1} for general $\sigma\in S_n$ follows by  iterative application of the above arguments for each transposition.

Consequently, 
\begin{align*}
    &P(D_j(\Delta)=k,\mathcal{A}^c)\geq   (1-o(k^{-\alpha}))\min_{\upsilon\in \mathcal{T}} \sum_{i=1}^m {\Delta \choose i}\frac{1}{\zeta(m,\alpha)i^\alpha}
 \prod_{t=1}^{\Delta}  P(\tau_j(t)=\overline{x}_t|\tau_j(t-1)=\overline{x}_{t-1},Y=\upsilon+i),
\end{align*}
where $\overline{x}_t=\tau_j(0)=i, t\in [\Delta-k]$ and $\overline{x}_t=\Delta+\tau_j(0)-t+i, t\in [\Delta-k+1,\Delta]$. So, 
\begin{align*}
    &P(D_j(\Delta)=k,\mathcal{A}^c)\geq
    (1-o(k^{-\alpha}))\min_{\upsilon\in \mathcal{T}}
    \sum_{i=1}^m {\Delta \choose k}\frac{1}{\zeta(m,\alpha)i^\alpha}
        \end{align*}
\begin{align*}
\\&\times \prod_{t=1}^{\Delta-k}  (1-\frac{i}{\upsilon+i+t-1})\times   \prod_{t=\Delta-k+1}^{\Delta}  \frac{i+t-\Delta+k-1}{\upsilon+i+t-1}
\\&= (1-o(k^{-\alpha}))
  \min_{\upsilon\in \mathcal{T}}
    \sum_{i=1}^m\frac{1}{\zeta(m,\alpha)i^\alpha}
\\&\times\frac{\Delta!}{k!(\Delta-k)!} \frac{(\upsilon+i-1)!}{(\upsilon+i+\Delta-1)!}
\frac{(\upsilon+\Delta-k-1)!}{(\upsilon-1)!}
\frac{(i+k-1)!}{(i-1)!}
\\& 
=(1-o(k^{-\alpha}))\min_{\upsilon\in \mathcal{T}}
    \sum_{i=1}^m \frac{1}{\zeta(m,\alpha)i^\alpha}
    \times \frac{{i+k-1 \choose i-1} { \upsilon+\Delta-k-1\choose \Delta-k}}
    {{\upsilon+i+\Delta-1 \choose \Delta}
  }
  \\& = 
  (1-o(k^{-\alpha}))\min_{\upsilon\in \mathcal{T}}
    \sum_{i=1}^m \frac{1}{\zeta(m,\alpha)i^\alpha}\times \frac{i}{i+k}
    \times \frac{{i+k \choose i} { \upsilon+\Delta-k-1\choose \upsilon-1}}
    {{\upsilon+\Delta+i-1 \choose \upsilon+i-1}
  }
      \end{align*}
Let $\psi\triangleq \frac{\upsilon^*}{\Delta}$, where $\upsilon^*$ achieves the minimum above, and define 
\[g(i)\triangleq \frac{{i+k \choose i}}{{(1+\psi)\Delta +i-1 \choose \psi \Delta+i-1}},\qquad  \psi k\leq i \leq \psi k +\sqrt{k}.\]
We first show that $g(i)$ is monotonically decreasing:
\begin{align*}
    \frac{g(i+1)}{g(i)}&= \frac{(i+k+1)(\psi \Delta+i)}{(i+1)((1+\psi)\Delta+i)}
     = \frac{k(\psi \Delta +i)+(i+1)\psi \Delta+i^2+i}{(i+1)\Delta+(i+1)\psi \Delta+i^2+i}
\end{align*}
It suffices to show that $
 k(\psi \Delta+i) \leq (i+1)\Delta$,
 which holds if and only if $\frac{i}{\Delta} \leq \frac{i-\psi k+1}{k}$ for  $n$ large enough. It can be verified that the latter holds since $k= o(n^{\frac{-\beta}{\alpha}})$ and $\Delta= \mu n$ and by noting that $\frac{\beta}{\alpha}\leq \frac{1}{3}$. Next, we show that $ \frac{g(k +\sqrt{k})}{g(k)}$ is bounded as $k\to \infty$. To see this, note that:
\begin{align}
& \nonumber  \frac{g(k +\sqrt{k})}{g(k)}=\prod_{j'=1}^{\sqrt{k}}\left(\frac{(1+\psi)k+j}{\psi k+j}
  \right)\left(\frac{\psi \Delta+\psi k+j-1}{(1+\psi)\Delta+\psi k+j-1}\right)
 \\&=\nonumber
 \prod_{j'=1}^{\sqrt{k}}\left(
 \frac{1+\psi}{\psi}- \frac{j}{\psi(\psi k+j)}
  \right)  \left(\frac{\psi}{1+\psi}+ \frac{\psi k+j-1}{(1+\psi)((1+\psi)\Delta+\psi k+j-1)}\right)
\\&=\nonumber
 \prod_{j'=1}^{\sqrt{k}}\left(
 1- \frac{j}{(1+\psi)(\psi k+j)}
  \right) \left(1+ \frac{\psi k+j-1}{\psi((1+\psi)\Delta+\psi k+j-1)}\right),
\end{align}
where the second term is $1+O(\frac{k}{\Delta})= 1+o(n^{-1-\frac{\beta}{\alpha}})$, so:
\begin{align*}
  &  \prod_{j'=1}^{\sqrt{k}} \left(1+ \frac{\psi k+j-1}{\psi((1+\psi)\Delta+\psi k+j-1)}\right)= 1+o(\sqrt{k}n^{-1-\frac{\beta}{\alpha}}) =1+o( n^{-1-\frac{3\beta}{2\alpha}})
\end{align*}
Also,
\begin{align*}
1 &    \geq \prod_{j'=1}^{\sqrt{k}}\left(
 1- \frac{j}{(1+\psi)(\psi k+j)}
  \right)
  \geq 
  \left(
 1- \frac{\sqrt{k}}{(1+\psi)(\psi k)}
  \right)^{\sqrt{k}}\geq e^{-\frac{1}{(1+\psi)\psi}},
\end{align*}
where in the last equality we have used the well-known result that $(1+\frac{a}{n})^n\to e^{a}, a>0$ as $n\to \infty$. Consequently, we have: 
\begin{align*}
& \nonumber  \frac{g(k +\sqrt{k})}{g(k)}=
c(1+o( n^{-1-\frac{3\beta}{2\alpha}}),
\end{align*}
for some constant $c\in [e^{-\frac{1}{(1+\psi)\psi}},1]$. Note that $\frac{\beta}{\alpha}\leq \frac{1}{3}$. So, $\frac{g(k +\sqrt{k})}{g(k)}$ is bounded as $n\to \infty$.
Hence, 
      \begin{align*}          &P(D_j(\Delta)=k,\mathcal{A}^c) \geq 
  (1-o(k^{-\alpha}))
    \sum_{i=\psi k}^{\psi k +\sqrt{k}} \frac{1}{\zeta(m,\alpha)i^\alpha}\times \frac{i}{i+k}
    \times \frac{{i+k \choose i} { (1+\psi)\Delta-k-1\choose \psi \Delta-1}}
    {{(1+\psi)\Delta+i-1 \choose \psi \Delta+i-1}
  }
 \\&\geq 
 c'\sqrt{k}(1-o(k^{-\alpha}))
    \frac{1}{\zeta(m,\alpha)(\psi k)^\alpha}\times 
    \frac{{(1+\psi)k \choose \psi k } { (1+\psi)\Delta-k-1\choose \psi \Delta-1}}
    {{(1+\psi)\Delta+\psi k-1 \choose \psi \Delta+\psi k-1}
  },
      \end{align*}
      where we have defined $c'\triangleq c \frac{\psi}{\psi+1}$. Next, we use the fact that ${n\choose k}= c_k \sqrt{\frac{n}{2\pi k (n-k)}}e^{nh_b(\frac{k}{n})}, n\in \mathbb{N}, k\leq n$, where $c_k\in [e^{\frac{-1}{6k}},1]$, and $h_b(p)=-p\ln(p)-(1-p)\ln(1-p)$ is the binary entropy function measured in nats (e.g. \cite{stanica2001good}) as follows:
      \begin{align*}
 &P(D_j(\Delta)=k,\mathcal{A}^c)
\geq 
 c'\sqrt{k}(1-o(k^{-\alpha}))
    \frac{1}{\zeta(m,\alpha)(\psi k)^\alpha}
    \frac{{(1+\psi)k \choose \psi k } { (1+\psi)\Delta-k-1\choose \psi \Delta-1}}
    {{(1+\psi)\Delta+\psi k-1 \choose \psi \Delta+\psi k-1}}
\\& = c'' (1-o(k^{-\alpha}))
    \frac{1}{\zeta(m,\alpha)(\psi k)^\alpha}\sqrt{\frac{((1+\psi)\Delta-k-1)(\psi \Delta+\psi k -1)\Delta}{ (\psi \Delta-1)(\Delta-k-1)((1+\psi)\Delta+\psi k -1) }}\times
\\& exp\Big((1+\psi) kh_b(\frac{1}{\psi+1})+((1+\psi)\Delta-k-1)h_b(\frac{\psi \Delta-1}{(1+\psi)\Delta-k-1})-
\\&((1+\psi)\Delta+\psi k -1) h_b(\frac{\psi\Delta+\psi k-1}{(1+\psi)\Delta+\psi k -1})\Big),
      \end{align*}
where $exp(x)=e^x, x\in \mathbb{R}$, and $c''\in \mathbb{R}$ is a constant number. We note that
\begin{align*}
    \sqrt{\frac{((1+\psi)\Delta-k-1)(\Delta+\psi k -1)\Delta}{ (\psi\Delta-1)(\Delta-k)((1+\psi)\Delta+\psi k -1) }}\to 1 \text{ as } n\to \infty,
\end{align*}
since $\Delta=\mu n$ grows linearly in $n$ by the sparsity condition in Section \ref{sec:gen}. Also, 
\begin{align*}
  & exp\Bigg(\left(1+\psi\right) kh_b\left(\frac{\psi}{\psi+1}\right)+\left((1+\psi)\Delta-k-1\right)h_b\left(\frac{\psi\Delta-1}{(1+\psi)\Delta-k-1}\right)-
\\&\left((1+\psi)\Delta+\psi k -1\right) h_b\left(\frac{\psi\Delta+\psi k-1}{(1+\psi)\Delta+\psi k -1}\right)\Bigg)
\\& = 
exp\Bigg(\left((1+\psi)\Delta-k-1\right)\left(h_b\left(\frac{\psi\Delta-1}{(1+\psi)\Delta-k-1}\right)-h_b\left(\frac{\psi}{\psi+1}\right)\right)+
\\&\left((1+\psi)\Delta+\psi k -1\right) \left(h_b\left(\frac{\psi}{\psi+1}\right)- h_b\left(\frac{\psi\Delta+\psi k-1}{(1+\psi)\Delta+\psi k -1}\right)\right)\Bigg).
\end{align*}
Using the second order Taylor's approximation, we have:
\begin{align*}
 & h_b\left(\frac{\psi\Delta-1}{(1+\psi)\Delta-k-1}\right)-h_b\left(\frac{\psi}{\psi+1}\right)= 
 \frac{-1+k\psi}{((1+\psi)\Delta-k-1)(\psi+1)}\log{\psi} -
  \\&\qquad \qquad 
 \frac{1}{2} \left(\frac{-1+k\psi}{((1+\psi)\Delta-k-1)(\psi+1)} \right)^2 \frac{1}{p_1(1-p_1)},
\end{align*}
where $p_1$ lies between $\frac{\psi\Delta-1}{(1+\psi)\Delta-k-1}$ and $\frac{\psi}{\psi+1}$ and is hence bounded away from 0 and 1 as $n\to \infty$. As a result $\frac{1}{p_1(1-p_1)}$ is bounded as $n\to \infty$.
Similarly,
\begin{align*}
  & \left(h_b\left(\frac{\psi}{\psi+1}\right)- h_b\left(\frac{\psi\Delta+\psi k-1}{(1+\psi)\Delta+\psi k -1}\right)
  \right)=
  \\&\frac{-\psi k +1}{(\psi+1)((1+\psi)\Delta+\psi k -1)}\log{\psi}-
  \\&\frac{1}{2}\left(\frac{-\psi k +1}{(\psi+1)((1+\psi)\Delta+\psi k -1)}\right)^2 \frac{1}{p_2(1-p_2)},
\end{align*}
where $p_2$ lies between $\frac{\psi\Delta+\psi k-1}{(1+\psi)\Delta+\psi k -1}$ and $\frac{1}{\psi+1}$ and is hence bounded away from zero as $n\to \infty$.
As a result,
\begin{align*}
&    exp\Bigg(\left(1+\psi\right) kh_b\left(\frac{1}{\psi+1}\right)+
\\& 
\left((1+\psi)\Delta-k-1\right)h_b\left(\frac{\psi\Delta-1}{(1+\psi)\Delta-k-1}\right)-
\\&\left((1+\psi)\Delta+\psi k -1\right) h_b\left(\frac{\psi\Delta+\psi k-1}{(1+\psi)\Delta+\psi k -1}\right)\Bigg)
\\&
= exp\Bigg(
 -\frac{1}{2} \left(\frac{(-1+k\psi)^2}{(\psi+1)^2((1+\psi)\Delta-k-1)} \right) \frac{1}{p_1(1-p_1)}
 \\& +
  \frac{1}{2}\left(\frac{(-\psi k +1)^2}{(\psi+1)^2((1+\psi)^2\Delta+\psi k -1)}\right) \frac{1}{p_2(1-p_2)}\Bigg)
\end{align*}
which is bounded as $n \to \infty$ since $\frac{k^2}{\Delta}= O(n^{1-\frac{2\beta}{\alpha}})$ as explained in the prequel. We have:
  \begin{align*}
 &P(D_j(\Delta)=k,\mathcal{A}^c)
 \geq \overline{c}(1-o(k^{-\alpha}))
    \frac{1}{\zeta(m,\alpha)(\psi k)^\alpha},
      \end{align*}
      for a constant $\overline{c}\in \mathbb{R}$ and $n$ large enough. A similar argument can be provided to show that
          \begin{align*}
 &P(D_j(\Delta)=k,\mathcal{A}^c)
 \leq \overline{c}'
    \frac{1}{\zeta(m,\alpha)(\psi k)^\alpha},
      \end{align*}
  For a constant $\overline{c}'$.    This completes the proof. \qed

\section{Proof of Proposition \ref{Prop:1}}
\label{Ap:Prop:1}
The following proves Equation \eqref{eq:P2:1}:
\begin{align*}
    \mathbb{E}(Y)&\stackrel{(a)}{=} \sum_{j=1}^n \mathbb{E}(\tau_j(0))
    \\& = n\mathbb{E}(\tau_1(0))= n\sum_{i=1}^m i\cdot \frac{1}{\zeta(m,\alpha)i^{\alpha}}\stackrel{(b)}{=} \frac{n\zeta(m,\alpha-1)}{\zeta(m,\alpha)},
\end{align*}
where in (a) we have used linearity of expectation, and in (b) we have used the definition of the Riemann Zeta function. Next, we prove Equation \eqref{eq:P2:2}:
\begin{align*}
    Var(Y)&\stackrel{(a)}{=} n Var(\tau_1(0))
    \\&\leq n\mathbb{E}(\tau^2_1(0))= n\sum_{i=1}^m \frac{1}{\zeta(m,\alpha)i^{\alpha}}i^2= \frac{n \zeta(m,\alpha-2)}{\zeta(m,\alpha)}. 
\end{align*}
where in (a) we have used independence of initial popularity values. 
To prove Equation \eqref{eq:P2:3}, note that by Chebychev's inequality, we have:
\begin{align*}
  & P(|Y-\mathbb{E}(Y)|>\epsilon \mathbb{E}(Y))\leq \frac{Var(Y)}{\epsilon^2\mathbb{E}^2(Y)}\leq  \frac{ \zeta(m,\alpha-2)\zeta(m,\alpha)}{n\epsilon^2 \zeta^2(m,\alpha-1)}.
\end{align*}
Furthermore, note that $f(x)=x^{-s}, s>0$ is a convex function. So, we have $\int_{x=1}^m f(x)dx \geq \sum_{i=1}^m f(i)$. Consequently,
\begin{align*}
    &\zeta(m,\alpha-t)\leq \int_{x=1}^m x^{t-\alpha}dx=\frac{m^{t+1-\alpha}-1}{t+1-\alpha}, t=0,1,2.
\end{align*}
So, 
\begin{align*}
 & P(|Y-\mathbb{E}(Y)|>\epsilon \mathbb{E}(Y))\leq 
 \frac{\frac{m^{1-\alpha}-1}{1-\alpha}\cdot\frac{m^{3-\alpha}-1}{3-\alpha}}{n\epsilon^2 \cdot\frac{m^{2-\alpha}-1}{2-\alpha}} = \frac{2-\alpha}{\epsilon^2(1-\alpha)(3-\alpha)}
 \frac{(m^{1-\alpha}-1)(m^{3-\alpha}-1)}{n(m^{2-\alpha}-1)}
 \\&\stackrel{(a)}{=}O(  \frac{max(1,m^{3-\alpha})}{n}) \stackrel{(b)}=O(n^{\beta}),
\end{align*}
where in (a) we have used the assumption that $\alpha>2$, and in (b) we have used the fact that $m=\eta n$ for a constant $\eta>0$.
\qed 

\section{Proof of proposition \ref{Prop:2}}
\label{Ap:Prop:2}
Note that $P(\tau_j(0)=1)\to 1$ as $\alpha\to \infty$; therefore, the summation of initial popularities is equal to the total number of groups, i.e. $Y=n$.  So,
\begin{align*}
 &  \lim_{\alpha\to \infty}P(D_j(\Delta)=k)=  \lim_{\alpha\to \infty}
  \sum_{b^{\Delta}\in \mathcal{B}_{\Delta,k}}\prod_{t=1}^{\Delta}  P(\tau_j(t)=x_t),
\end{align*}

Similar to the steps in Appendix \ref{Ap:th:1}, we have:

\begin{align*}
    &\lim_{\alpha \to \infty} P(D_j(\Delta)=k)=
    {\Delta \choose k}\prod_{t=1}^{\Delta-k}  (1-\frac{1}{n+t-1})\times   \prod_{t=\Delta-k+1}^{\Delta}  \frac{t-\Delta+k}{n+t-1}
\\&=
    \frac{\Delta!}{k!(\Delta-k)!} \frac{(n-1)!}{(n+\Delta-1)!}
\frac{(n+\Delta-k-2)!}{(n-2)!}\frac{k!}{0!}
\\& =
     \frac{{n+\Delta-k-2 \choose n-2}}{{\Delta+n-1 \choose \Delta}}
     \end{align*}
     Recall that $\Delta=(1+\mu)n$. As a result,
     \begin{align*}
     &\lim_{\alpha \to \infty} P(D_j(\Delta)=k)=   \frac{{(1+\mu)n-k-1 \choose n-2}}{{(1+\mu)n \choose \mu n}}
  \\&\stackrel{(a)}{=}
  c\sqrt{\frac{((1+\mu)n-k-1)\mu n}{(n-2)(\mu n-k+1)(1+\mu)}}\times
  \\&
  exp((1+\mu)n-k-1)h_b(\frac{n-2}{(1+\mu)n-k-1})-
  \\&
  (1+\mu)nh_b(\frac{\mu }{1+\mu}))
\end{align*}
where in (a) we use the fact that ${n\choose k}= c_k \sqrt{\frac{n}{2\pi k (n-k)}}e^{nh_b(\frac{k}{n})}, n\in \mathbb{N}, k\leq n$, where $c_k\in [e^{\frac{-1}{6k}},1]$.
Note that
\begin{align*}
    & \sqrt{\frac{((1+\mu)n-k-1)\mu n}{(n-2)(\mu n-k+1)(1+\mu)}} \to 1 \text{ as } n \to \infty.
\end{align*}
So, 
\begin{align*}
     &\lim_{n\to \infty}\lim_{\alpha \to \infty}P(D_j(\Delta)=k)=\lim_{n\to \infty}  c \times 
  exp((1+\mu)n-k-1)
  h_b(\frac{n-2}{(1+\mu)n-k-1})-
  (1+\mu)nh_b(\frac{1}{1+\mu}))
\end{align*}
Using the second order Taylor's approximation, we have:
\begin{align*}
  & (1+\mu) n(h_b(\frac{n-2}{(1+\mu)n-k-1})-
  h_b(\frac{1}{1+\mu}))=
  \\& (1+\mu)n\Big(\frac{k-1-2\mu }{(1+\mu)(1+\mu n-k-1)} \log(\mu)-
  \\&
  \frac{1}{2}\frac{(k-1-2\mu)^2}{(1+\mu)^2((1+\mu) n-k-1)^2}\frac{1}{p_1(1-p_1)}\Big)
  \\&\stackrel{(a)}{=}
  \frac{k}{1+\mu}\log(\mu)+O(1)
\end{align*}
where in (a) $p_1$ lies between $\frac{n-2}{(1+\mu)n-k-1}$ and $\frac{1}{1+\mu}$ and is hence bounded away from 0 and 1 as $n\to \infty$. As a result $\frac{1}{p_1(1-p_1)}$ is bounded as $n\to \infty$. Furthermore,
\begin{align*}
    & exp\Big((1+\mu)n-k-1)h_b(\frac{n-2}{(1+\mu)n-k-1})-
    \\&
  (1+\mu)nh_b(\frac{1}{1+\mu})\Big)=
  \\&
  exp\Big(\frac{k}{1+\mu}\log(\mu)-(k+1)h_b(\frac{1}{1+\mu})+O(1)\Big)=
  \\&
  exp\Big(\frac{k}{1+\mu}\log(\mu)-\frac{k}{1+\mu}\log(1+\mu)-\frac{k\mu}{1+\mu}\log(\frac{1+\mu}{\mu})-
  \\&
  h_b(\frac{1}{1+\mu})+O(1)\Big)=
  \\&
  exp(k\log(\mu)-k\log(1+\mu)-h_b(\frac{1}{1+\mu})+O(1))
  \\&
  exp(\log(\frac{\mu}{1+\mu})^k+O(1))=
  c(\frac{\mu}{1+\mu})^k
\end{align*}
So,
\begin{align*}
     &\lim_{n\to \infty}\lim_{\alpha \to \infty}P(D_j(\Delta)=k)=c(\frac{\mu}{1+\mu})^k
\end{align*}
Note that we must have $c=\frac{1}{1+\mu}$. This completes the proof. \qed
\section{Proof of Proposition \ref{Prop:3}}
\label{App:Prop:3}
We provide an outline of the proof. Equation \eqref{eq:prop1:1} follows by linearity of expectation and the fact that $\sum_{j\in [n]}\mathbb{E}(D_{\Delta,j})=\Delta$. To show Equation \eqref{eq:prop1:2}, we note that $\mathbb{E}(D^2_j(\Delta))= \sum_{i=1}^m P(\tau_j(0)=i)\mathbb{E}(D^2_j(\Delta)|\tau_j(0)=i)$. Next, for a given $i\in [m]$ we construct a new  bipartite graph by replacing the right-vertex $v_{2,j}$ by $i$  right-vertices $v_{2,j,k}, k\in [i]$ each with initial popularity values $\tau_{j,k}(0)=1, k\in [i]$. It is straightforward to verify that the degree distribution of the right-vertices in the original graph, other than $v_{2,j}$, is the same as the new graph, and the degree distribution of $v_{2,j}$ in the original graph is the same as the sum of the degrees of the new vertices $v_{2,j,k}, k\in [i]$. So, $\mathbb{E}(D^2_j(\Delta)|\tau_j(0)=i)=\mathbb{E}((D_{j,1}(\Delta)+D_{j,2}(\Delta)+\cdots+ D_{j,i}(\Delta))^2|\tau_{j,k}(0)=1, k\in [i])$. Consequently, $\mathbb{E}(D^2_j(\Delta)|\tau_j(0)=i)\leq i \mathbb{E}(D_{j,1}^2(\Delta)|\tau_{j,1}(0)=1)+i(i-1)\mathbb{E}(D_{j,1}(\Delta)D_{j,2}(\Delta)|\tau_{j,1}=\tau_{j,2}=1)$. So, $\mathbb{E}(D_j^2(\Delta))\leq \mathbb{E}(\tau_j(0))\mathbb{E}(D^2_{j,1}(\Delta)|\tau_{j,1}(0)=1)+\mathbb{E}(\tau^2_j(0))\mathbb{E}(D_{j,1}(\Delta)D_{j,2}(\Delta)|\tau_{j,1}(0)=\tau_{j,2}=1)$. The two terms on the right hand side of the last equation are finite as $m\to \infty$ since $\mathbb{E}(\tau^2_j(0))=\zeta(m,\alpha-2)$, $\mathbb{E}(\tau_j(0))=\zeta(m,\alpha-1)$, and due to Proposition 1 in \cite{shariatnasab2021fundamental} which shows the result conditioned on $\tau_j(0)=1$. To prove Equation \eqref{eq:prop1:3}, we have:
 \begin{align*}
    \mathbb{E}(D_{\Delta,i}D_{\Delta,j})
    &=   \mathbb{E}(D_{\Delta,1}(\frac{1}{n-1}\sum_{j'=2}^{n}D_{\Delta,j'}))
    \\&=\frac{\Delta}{n-1}\mathbb{E}(D_{\Delta,1})-\frac{1}{n-1}\mathbb{E}(D^2_{\Delta,1})
= \mu^2+O(\frac{1}{n}).
\end{align*}
Equations  \eqref{eq:prop1:2.5}, \eqref{eq:prop1:4}, and  \eqref{eq:prop1:5} can be shown similarly. The proof is omitted for brevity.   
\section{Proof of Proposition \ref{Prop:5}}
\label{App:Prop:5}
We provide an outline of the proof.  Note that $\frac{C_i}{n}= \frac{1}{n}\sum_{j=1}^n \mathbbm{1}(R(i,j))$, where $R_{i,j}=\mathbbm{1}(v_{1,i}\in \mathcal{V}_{2,j}), i\in [m], j\in [n]$. Also,  for any $\mathcal{A}\subset [n]$ we have:
\begin{align*}
    &\mathbb{E}(\prod_{j\in \mathcal{A}}(R_{i,j}))= P(R_{i,j}=1, j\in \mathcal{A})
\\&= \sum_{\substack{d^{\mathcal{A}},\\i_{j}\in [m],j \in \mathcal{A}, \theta \in [m]}} P_{\Theta, (\tau_j(0), D_j)_{j\in \mathcal{A}}}(\theta,d^{\mathcal{A}}, i^{\mathcal{A}})\prod_{j\in \mathcal{A}} P(R_{i,j}=1|\theta,d_j,i_j)
\\&=  \sum_{\substack{d^{\mathcal{A}},\\i_{j}\in [m],j \in \mathcal{A}, \theta \in [m]}} P_{\Theta, (\tau_j(0), D_j)_{j\in \mathcal{A}}}(\theta,d^{\mathcal{A}},i^{\mathcal{A}})\\&\qquad \times \prod_{j\in \mathcal{A}} \frac{d_j+i_j}{m+\theta}
{\leq } \frac{1}{m^{|\mathcal{A}|}}\mathbb{E}(\prod_{j\in \mathcal{A}}(D_j+\tau_j(0))
  {\leq} \frac{\lambda(m,\alpha)}{m}^{|\mathcal{A}|}.
\end{align*}
So, using an extension of Hoeffding's inequality to weakly correlated variables given in \cite{impagliazzo2010constructive}, we have:
\begin{align*}
   & P(C_i\geq \ell)\leq c2^{-nD_b(\frac{\lambda(m,\alpha)}{m}(1+\psi)||\frac{\lambda(m,\alpha)}{m})},
\end{align*}
where $\ell= \frac{n}{m}\lambda(m,\alpha)(1+\psi)=\frac{1}{\beta}\lambda(m,\alpha)(1+\psi)$ and $\psi\in (0, \frac{\lambda(m,\alpha)}{\mu}-1)$. To derive \eqref{eq:prop2}, we note that: 
\begin{align*}
   & P(C_i\geq \ell)\leq c2^{-n(\frac{\lambda(m,\alpha)}{m}(1+\psi)\log{(1+\psi)}+O(\frac{1}{n}))}\to 0,\text{ as } n\to \infty.
\end{align*}
\qed

\section{Proof of Theorem \ref{th:2}}
\label{App:th:2}
We provide an outline of the proof. Let $I_{k}(\mathcal{G}_s,Y^\kappa)$ be the information value of user $k\in [m]$ given scannd graph $\mathcal{G}_s$ and query responses $Y^{\kappa}, \kappa\in [n]$.
 Define the following stopping times
\begin{align*}
    &t_k\triangleq\min_{t}\bigg\{t\big|  I_{k}(\mathcal{G}_s,Y^t)>
    \log{\frac{1}{\epsilon}}\bigg\}, k\in [m],
   \qquad \qquad t^*\triangleq\min_{k\in [m]}t_k
\end{align*}
Note that $\overline{Q}_{A\text{-}ITS}=\mathbb{E}(t^*)$.
Fix $n'\in \mathbb{N}$. Let $T_{n'}= \min\{t_M,n'\}$. Note that: 
\begin{align}
    \mathbb{E}(I_{T_{n'}(M)})
    \geq c'(\sum_{d'\geq d}\mathbb{E}(N_{d})I_{d,\theta}(Y:E_0)+
      i'_{\theta}I_{d-1}(Y;E_0))-H(M),
\label{eq:App:F:1}\end{align}
where $\mathbb{E}(T_{n'})=\sum_{d'\geq d} \mathbb{E}(N_d')+i', i_{\theta}'\leq  \mathbb{E}(N_{d'-1}), \mathbb{E}(N_d')=\frac{n}{\zeta(m,\alpha) d'^\alpha}$, and we have used Wald's identity \cite{wald1944cumulative} and   Proposition \ref{Prop:5} to upper bound the expectation over the fingerprint distribution with that over a product distribution.
Note that 
\begin{align}
 \mathbb{E}\left(I_{T_{n'}}\left(M\right)\right)\leq \mathbb{E}\left(I_{T_{n'}-1}\left(M\right)\right)+i_{max}\leq \log\frac{1}{\epsilon}+i_{max}.
 \label{eq:App:F:2}
\end{align}
Equations \eqref{eq:App:F:1} and \eqref{eq:App:F:2} yield the desired bound on $\overline{Q}_{A\text{-}ITS}$. The proof  for the probability of error follows similar steps as that of Theorem 1 in \cite{shariatnasab2021fundamental} and is provided for completeness as follows:
\begin{align*}
  &  P_e= P(\exists j\neq M: t_j\leq t_M)\leq \sum_{j\neq M} P(t_j\leq \infty)
  = \sum_{j\neq M} \lim_{\eta\to \infty} P(\kappa_j\leq \eta)
  \\&\stackrel{(a)}{=}\sum_{j\neq M} \lim_{\eta\to \infty} \mathbb{E}_{P_{Y^n,(R_{M,i})_{i\in [n]}}}\left(\frac{P_{Y^n}P_{(R_{M,i})_{i\in [n]}}}{P_{Y^n,(R_{M,i})_{i\in [n]}}}\mathbbm{1}(\kappa_j\leq \eta))\right)
  \\&\leq \sum_{j\neq M} \lim_{\eta\to \infty} 
  \frac{(1+o(1))}{c'}\mathbb{E}_{P_{Y_i,R_{M,i}}}\left(\prod_{i\in [n]}\frac{P_{Y_i}P_{R_{M,i}}}{P_{Y_i,R_{M,i}}}\mathbbm{1}(\kappa_j\leq \eta))\right)
  \\&\leq 
  \sum_{j\neq M} \lim_{\eta\to \infty} 
  \frac{(1+o(1))}{c'}\mathbb{E}_{P_{Y_i,R_{M,i}}}\left(e^{-\log\frac{1}{\epsilon}-I_0(M)})\right)
  \\&=  \sum_{j\neq M} \frac{1}{c'}\epsilon P_M(j)\leq \frac{1}{c'}\epsilon(1+o(1)).
\end{align*}
where in (a) we have used the fact that $P_{(R_{j,i})_{i\in [n]}}= P_{(R_{M,i})_{i\in [n]}}, j\in [m]$. 

\end{appendices}

 \clearpage

\bibliographystyle{unsrt}
\bibliography{References}

\end{document}